\documentclass[11pt]{article} 

\usepackage{fullpage,amsmath,xspace,amssymb,epsfig}
\usepackage{color}
\usepackage{bbm}
\usepackage{amsthm}
\usepackage{hyperref}

\newtheorem{theorem}{Theorem}
\newtheorem{lemma}[theorem]{Lemma}
\newtheorem{corollary}[theorem]{Corollary}
\newtheorem{proposition}[theorem]{Proposition}
\newtheorem{claim}[theorem]{Claim}
\newtheorem{conjecture}[theorem]{Conjecture}
\newtheorem{definition}{Definition}
\newtheorem{fact}[theorem]{Fact}

\theoremstyle{definition}
\newtheorem{remark}[theorem]{Remark}

\newcommand{\ignore}[1]{}

\newcommand\mbR{\mbox{$\mathbb{R}$}}

\newcommand\Q{\mbox{$\mathbb{Q}$}}

\newcommand\D{\mbox{\sf {D}}\xspace}

\newcommand\dcc{\mbox{$\sf {D^{CC}}$}\xspace}

\newcommand\gf{\mbox{$\mathbb{F}_2$}\xspace}

\newcommand\B{\{0,1\}}     
\newcommand\pmB{\{+1,-1\}}     
\newcommand\Bn{\{0,1\}^n}
\newcommand\BntB{\{0,1\}^n\rightarrow \{0,1\}}
\newcommand {\ie} {\textit{i.e.}\xspace}

\newcommand {\etal} {\textit{et al.}\xspace}
\newcommand\defeq{\stackrel{\mathrm{\scriptsize def}}{=}}

\newcommand\av{\mbox{\bf{\bf E}}}

\newcommand\alice{\mbox{\sf {Alice}}\xspace}
\newcommand\bob{\mbox{\sf {Bob}}\xspace}

\newcommand\rank{\mbox{\sf {rank}}\xspace}
\newcommand\gran{\mbox{\sf {gran}}\xspace}
\newcommand\spar{\mbox{\sf {spar}}\xspace}

\newcommand\codim{\mbox{\sf {co-dim}}\xspace}

\newcommand\fn[2]{\| \hat{#1} \|_#2}
\newcommand\wfn[2]{\| \widehat{#1} \|_#2}

\newcommand\lrank{\mbox{\sf {lin-rank}\xspace}}

\newcommand{\cube}[1]{{\operatorname\{0, 1\}^{#1}}}

\newcommand{\eqdef}{{\stackrel{\rm def}{=}}}

\newcommand\floor[1]{\lfloor #1 \rfloor}
\newcommand\eqd{\equiv_d}

\newcommand{\polylog}[1]{\mathrm{polylog}{#1}}

\newcommand\red[1]{\textcolor{red}{#1}}
			
\begin{document}
\title{Fourier Sparsity of GF(2) Polynomials}
\vspace{1em}
\author{Hing Yin Tsang\thanks{University of Chicago, Chicago, IL 60637, USA. Email: {\texttt hytsang@uchicago.edu}} 
\and Ning Xie\thanks{Florida International University, Miami, FL 33199, USA. Email: {\texttt nxie@cis.fiu.edu}} 
\and Shengyu Zhang\thanks{The Chinese University of Hong Kong, Shatin, NT, Hong Kong. Email: {\texttt syzhang@cse.cuhk.edu.hk}}}
\date{}
\setcounter{page}{0}
\maketitle

\begin{abstract}
We study a conjecture called ``linear rank conjecture'' recently raised in (Tsang \etal, FOCS'13), 
which asserts that if many linear constraints are required to lower the degree of a GF(2) polynomial, 
then the Fourier sparsity (i.e. number of non-zero Fourier coefficients) of the polynomial must be large. 
We notice that the conjecture implies a surprising phenomenon that 
if the highest degree monomials of a GF(2) polynomial satisfy a certain condition,
then the Fourier sparsity of the polynomial is large regardless of the monomials of lower degrees -- 
whose number is generally much larger than that of the highest degree monomials.  
We develop a new technique for proving lower bound on the Fourier sparsity of GF(2) polynomials, 
and apply it to certain special classes of polynomials to showcase the above phenomenon. 
\end{abstract}

\section{Introduction}
The study of \emph{communication complexity}, introduced by Yao~\cite{Yao79} in 1979,
aims at investigating the minimum amount of information exchange required for computing functions whose
inputs are distributed among multiple parties~\cite{KN97}.
In the standard two-party setting, \alice holds an input $x$, \bob holds an input $y$, and they 
wish to compute a function $F$ on $(x,y)$ by as little communication as possible.
Perhaps the most important open problem in communication complexity is the so-called \emph{Log-rank Conjecture}
proposed by Lov{\'a}sz and Saks~\cite{LS88},
which states that the \emph{deterministic communication complexity} of any $F: \{0,1\}^{n} \times \{0,1\}^{n} \to \{0,1\}$, $\dcc(F)$, 
is upper bounded by a polynomial of the logarithm of the rank the communication matrix $M_F = [F(x,y)]_{x,y}$,
where the rank is taken over the reals.
Although a lot of effort has been devoted to the conjecture in the past two decades, 
very little progress has been achieved and the best upper bound known to date is 
$\dcc(F)=O\left(\sqrt{\rank(M_F)}\log \left(\rank(M_F)\right)\right)$, due to Lovett~\cite{Lov14}. 
Note that there is still an exponential gap between this and the best known lower bound, which is
$\dcc(F)=\Omega\left((\log\rank(M_F))^{\log_{3} 6}\right)$ due to Kushilevitz (unpublished, cf.~\cite{NW95}).
For an overview of recent developments in this direction, see~\cite{Lov14a}.

An interesting special class of functions computable by two parties is the so-called \emph{XOR functions}.
Specifically, $F$ is an XOR function if there exists an $f: \{0,1\}^{n} \to \{0,1\}$ such that
for all $x$ and $y$, $F(x,y) = f(x\oplus y)$, where $\oplus$ is the bit-wise XOR. 
Denote such $F$ by $f\circ \oplus$.
Besides including important examples such as Equality and Hamming Distance, 
XOR functions are particularly interesting for studying the Log-rank Conjecture 
due to its intimate connection with the analysis of Boolean functions.
Specifically, if $F$ is an XOR function, 
then the rank of $M_{F}$ is just the Fourier sparsity of $f$ (i.e., the number of non-zero Fourier coefficients of $f$)~\cite{BC99}. 
Therefore proving the Log-rank conjecture for XOR functions can be achieved by demonstrating 
short \emph{parity decision tree} protocols\footnote{Recall that a \emph{parity decision tree} $T$ for a function $f: \{0,1\}^{n} \to \{0,1\}$ generalizes an ordinary decision tree in the sense that each internal node of $T$ is now associated with a linear function $\ell(x)$, instead of a single bit, of the input, and $T$ branches according to the parity of $\ell(x)$.} computing Fourier sparse Boolean functions, 
and this problem attracted a lot of attention~\cite{ZS09,LZ10,MO09,TWXZ13,STV14} during the past years.

Recently, by viewing Boolean functions as $\gf$-polynomials, a new communication protocol based on $\gf$-degree reduction 
was proposed in~\cite{TWXZ13} for XOR functions:
suppose $f(x\oplus y)$ is a degree-$d$ polynomial and $r_d$ is the minimum number of variables (up to an invertible linear transformation) 
restricting of which reduces $f$'s degree to at most $d-1$, 
then Alice and Bob both apply the optimal linear map to their inputs and
send each other $r_d$ bits of their respective inputs.
Repeating this process at most $d-1$ times, the restricted function of $f$ becomes a constant function hence 
they successfully compute $f(x\oplus y)$.
Of course, such a protocol is efficient only if the numbers $r_d, r_{d-1}, \ldots, r_1$,  
of the restricted variables that they need to exchange, are not large. 
Studying these quantities, namely \emph{linear ranks} of polynomials, is one the central objectives of this paper.

\begin{definition}[linear rank of a polynomial]\label{def:linear_rank}
Let $f$ be a degree-$d$ polynomial, $V$ be a subspace in $\cube{n}$ and $H=a+V$ be any affine shift of $V$. 
Denote by $f|_H$ the restriction of $f$ on $H$. Then the \emph{linear rank} of $f$, 
denoted $\lrank(f)$, is the minimum co-dimension of any subspace $H$ such that the degree of $f|_H$ is strictly less than $d$;
that is,
\[
\lrank(f)=\min_{\deg_{2}(f|_{H})<\deg_{2}(f)}\codim(H). 
\]
\end{definition}

In other words, $\lrank(f)$ is the minimum number of linear functions one needs to fix 
in order to lower the degree of $f$. 
Consider, for example, the degree-$3$ polynomial 
$f(x_1, \ldots, x_{3n})=(x_1+\cdots+x_n)(x_{n+1}+\cdots+x_{2n})(x_{2n+1}+\cdots+x_{3n})$.
In the original basis, one needs to fix at least $n$ variables to lower the degree of $f$.
However, fixing one linear function $x_1 + \cdots + x_n = 0$ is
enough to lower its degree. Therefore $\lrank(f)=1$.  

For a Boolean function $f$, let $\spar(f)$ denote the Fourier sparsity of $f$ and 
$\D_\oplus(f)$ denote the parity decision tree complexity of $f$. 
As restrictions do not increase $\spar(f)$ (cf. Lemma~\ref{lem:rotation}) and  $\deg_{2}(f)\leq \log{\spar(f)}$ for every $f$, 
the following \emph{linear rank conjecture}---if true---would readily implies the
Log-rank Conjecture for XOR functions.

\begin{conjecture}[Linear rank conjecture~\cite{TWXZ13}]\label{conj:lin-rank}
For any $f:\BntB$, the linear rank of $f$ is upper bounded by polylogarithmic of the Fourier sparsity of $f$: 
$\lrank(f) = O(\log^c(\spar(f)))$ for some $c=O(1)$. 
Equivalently, if $\lrank(f)=r$, then $\spar(f)=2^{r^{\Omega(1)}}$.
\end{conjecture}
Although it is still open whether the linear rank conjecture is equivalent to the Log-rank Conjecture for XOR functions,
it is worthwhile to note that it is equivalent to the stronger statement that
$\D_\oplus(f)=\polylog(\spar(f))$ for any Boolean function $f$.

\subsection{Large Fourier sparsity determined by highest degree monomials only} 
Before further discussing the linear rank conjecture, let us first state a lemma of~\cite{TWXZ13} (Lemma 19) 
in a slightly stronger form and give an alternative simple proof (another simple proof used polynomial derivatives ~\cite{CT13}). 
The lemma says that, once the linear subspace $V$ in Definition~\ref{def:linear_rank} is identified, 
it does not matter which affine shift is used in the definition
of linear rank: all affine subspaces of $V$ are equally good.
More specifically, if $f$ restricted to $a+V$ has degree at most $d-1 $ (where $d = \deg_2(f)$), 
then $f$ restricted to any other $a'+V$ also has degree at most $d-1$. 
This can be seen by the following argument. Call a monomial in $f$ a \emph{maxonomial} if it is of the maximal degree (i.e., degree $d$). 
Apply a linear map to $\cube{n}$ so that $V = \{x: x_1 = \cdots = x_r = 0\}$, where $r = \codim(V)$. 
Then $f|_{a+V}$ becomes a polynomial of degree at most $d-1$ if and only if every maxonomial of $f$ (under the new basis) contains
at least one variable in the set $\{x_1, \ldots, x_r\}$. 
Moreover, when this happens it does not matter whether $x_i$ ($i\le r$) is restricted to $0$ or $1$, the degree of the maxonomial always decreases, 
thus $\deg_2(f|_{a'+V}) \le d-1$ for all $a'\in \Bn$.

The above fact also reveals that 
the linear rank $r$ of any polynomial $f(x)$ is determined by the maxonomials in $f(x)$ \emph{only}. 
Fourier sparsity in general, on the other hand, should depend on all GF(2) monomials, not only those with the highest degree. 
However, the linear rank conjecture claims that if the maxonomials in $f(x)$ make the linear rank large,
then no matter how the lower-degree monomials behave, the Fourier sparsity is large.
Therefore, for the effect of forcing the Fourier sparsity of GF(2) polynomial to be large, 
there exists a surprising fact (assuming the linear rank conjecture) 
that can be summarized by paraphrasing a famous quote from \emph{Animal Farm}: 
``All monomials are equal, but some monomials are more equal than others''.

In retrospect, this phenomenon is known for some extremal cases. When $\deg_2(f) = 2$, 
the lower degree terms form a linear function $\chi_\alpha$, 
adding which only shifts Fourier spectrum by $\alpha$ and thus does not affect the Fourier sparsity. 
When $\deg_2(f) = n$, the Fourier sparsity is at least $2^{\deg_2(f)}-1 = 2^n-1$, 
which is again determined by the (unique) maxonomial. But for general $2<d<n$, 
maxonomials by themselves do not necessarily determine large Fourier sparsity. 
For instance, if there is only one maxonomial $x_1\ldots x_d$, 
then the Fourier sparsity can be as small as $2^d$ (when, say, the lower degree part is $x_1+\cdots +x_n$), 
and as large as $2^{n-d}$ (when, say, the lower degree part is a bent function\footnote{A Boolean function
$f:\{0,1\}^{m} \to \{-1,1\}$ is \emph{bent} if its Fourier coefficients satisfy that $|\hat{f}(\alpha)|=2^{-m/2}$
for all $\alpha \in \{0,1\}^{m}$.}
over $x_{d+1}$, \ldots, $x_n$). 
Despite this uncertainty, we will show that when the maxonomials form certain patterns, 
the Fourier sparsity is guaranteed to be large, regardless of the lower degree terms (
whose number can be much larger than that of maxonomials). 
One sufficient condition for the pattern is that the linear rank, which depends on maxonomials only, is large. 
And we will showcase some specific classes of good patterns.

Therefore, apart from leading directly to a proof of the Log-rank Conjecture for XOR functions, 
studying the linear rank conjecture is interesting in its own right, 
due to its close connection to the Fourier analysis of Boolean functions in the GF(2) polynomial representation.

\subsection{Our work}
We study the linear rank conjecture and in particular investigate how could the maxonomials of a $\gf$-polynomial possibly 
determine by themselves the Fourier sparsity of the polynomial. 
We develop a new technique which is able to show that, under certain circumstances, the Fourier sparsity is large
for all possible settings of lower degree monomials. 
It is hoped that this new framework of studying the Fourier coefficients based on 
GF(2) monomials
may be further extended and generalized to yield more structural results on the analysis of Boolean functions,
such as sparsity, granularity and Fourier mass distribution.

For general degree-$d$ polynomials, we investigate the linear rank and Fourier sparsity for several special cases. 
Since the maxonomials of a polynomial are the main concern of the conjecture,
it is convenient to borrow the terminology of hypergraphs to define these maxonomials. 
For example, the complete $d$-uniform maxonomials corresponds to the degree-$d$ polynomial who has all $\binom{n}{d}$ maxonomials.

\subsubsection{Linear rank of polynomials with complete $d$-uniform maxonomials}
We determine the exact values of the linear ranks of degree-$d$ polynomials with all $\binom{n}{d}$ maxonomials.
Specifically, let $f=\sum_{\text{$S\subset[n]$, $|S|=d$}}\prod_{i\in S}x_i + f'$, 
where $f'$ is an arbitrary polynomial of degree at most $d-1$, we show that for such an $f$, 
\[
\lrank(f) = 
\begin{cases}
\floor{\frac{n}{2}}-\frac{d}{2}+1 & \text{if $d$ is even},\\
1 & \text{if $d$ is odd}.
\end{cases}
\]
The proof exploits the symmetry of maxonomials and goes through a careful induction on $n$ and $d$.
In particular we prove a ``step-function'' type behaviour of the linear rank (for fixed $d$ and with respect to $n$),
by showing both upper and lower bounds for the number of linear functions one needs to fix in order to decrease the degree of the polynomial.

\subsubsection{Fourier sparsity of polynomials with complete $d$-uniform maxonomials}
If the linear rank conjecture is true, then for any polynomial with complete $d$-uniform maxonomials ($d$ is even), 
the Fourier sparsity must be $2^{n^{\Omega(1)}}$ regardless of the lower degree monomials. 
We are only able to verify this for a small (but infinite) set of $d$'s:
for any $d$ that is a power of $2$, if $f:\cube{n}\to \{0,1\}$ is a degree-$d$ polynomial with complete $d$-uniform maxonomials, then
\[
\spar(f)\geq 2^{d\cdot\lfloor n/d\rfloor}-1=\Omega(2^{n}).
\] 
We prove this sparsity lower bound by developing a new technique to be discussed more later.
Zhang and Shi~\cite{ZS09} proved that any symmetric boolean function has Fourier sparsity $2^{\Omega(n)}$, 
unless it is constant, the parity function over $n$ bits or its negation. 
However, as the polynomials considered there are symmetric, 
their result requires the degree-$d'$ monomials to be either empty or complete $d'$-uniform, for every $d' \leq d$.
On the contrary, our lower bound applies to a broader class of functions as it holds for all possible choices of lower degree monomials, 
as long as the highest-degree monomials are symmetric.

\subsubsection{Other results}
We further demonstrate the power of our technique by applying it to several other special forms of sparse maxonomials.
In particular, we show lower bounds on the Fourier sparsity of polynomials whose maxonomials are pairwise disjoint
or have certain ``regular'' overlaps.    

Gopalan \etal~\cite{GOS+11} studied the \emph{granularity} of a function's Fourier spectrum, which is the smallest 
integer $k$ such that all Fourier coefficients of the function can be expressed as integer multiples of $1/2^k$. 
They showed that for any Boolean function $f:\cube{n}\to \{0,1\}$, $\gran(f) \leq \log \spar(f)$.
On the other hand, by Parseval's identity, $\log \spar(f) \leq  2\gran(f)$. 
The granularity of a linear functions is $1$ and the maximum granularity of any $n$-variate quadratic
polynomial is $n/2$. It thus natural to conjecture that, for any $n$-variate low-degree polynomial $f(x)$, 
although $\spar(f)$ can be as large as $2^n$, the granularity of $f(x)$ is always bounded away from $n$.
We are able to apply our technique to show the following \emph{upper bound} on the granularity of low-degree polynomials: 
for any degree-$d$ polynomial $f$, $\gran(f) \leq n-\lceil \frac{n}{d} \rceil+1$. It is easy to see this bound is tight
as it is attained by the ``generalized inner product function'': $f(x)=x_1 x_2 \cdots x_d + \cdots + x_{(k-1)d+1}x_{(k-1)d+2}\cdots x_{kd}$,
where $n=kd$.

\subsubsection{Techniques}
The main challenge in proving sparsity lower bounds based on \emph{only} the maxonomials of a polynomial is how to isolate the effect
of \emph{all} lower degree monomials.
To the best of our knowledge, there is no prior method or result of this kind.
Our method is to first apply the standard procedure to transform a degree-$d$ polynomial $f$ into a Fourier polynomial,
and then define a ``weight function'' $w_{f}(T)$ on each set $T\subseteq[n]$
such that the Fourier coefficient of $f$ at any set $S$ can be written as $\sum_{T\supseteq S}w_{f}(T)$.
This implies that the weight function at $[n]$ is the most important term as it contributes to all the Fourier coefficients of $f$.
Another nice property of the weight function is that for any $T$, $2^{|T|}w_{f}(T)$ can be expressed as a sum of
alternating terms in which the $k^{\text{th}}$ term is $(-2)^{k}N_{k}(T)$,
where $N_{k}(T)$ is the number of ways to cover $T$ with (the supports of) exactly $k$ monomials of $f(x)$.
Therefore, the problem of computing the Fourier coefficients of an $\gf$-polynomial is now reduced to a combinatorial problem of counting 
the numbers of covers of all subsets of $[n]$ using various numbers of sets from the set family defined by the monomials of the polynomial.
Moreover, the parity of $2^{|T|}w_{f}(T)$ is likely to be determined by the numbers of smaller covers due the factor $(-2)^{k}$
in each term of the sum.  
Using the notion of ``granularity'' introduced in~\cite{GOS+11}, our strategy for showing sparsity lower bound is to
argue that $w_{f}([n])$ is the single one with the highest granularity among all weight function values.
Note that if $n=kd$ and we can cover $[n]$ with (the supports of) maxonomials of $f(x)$ only, then these covers 
would be the minimum covers as they require only $k=n/d$ sets while any cover involving lower monomials is of size at least $k+1$.
Hence to prove that $w_{f}([n])$ has the highest possible granularity, 
it suffices to show that the number of $k$-covers of $[n]$ is odd, as we did for the several sparsity lower bounds.

\subsection{Organization of the paper}

Section~\ref{sec:prelim} contains notations and preliminaries that will be used throughout the paper.
In Section~\ref{sec:lrank} we compute exactly the linear rank of polynomials with complete $d$-uniform maxonomials.
The basic machinery for proving sparsity lower bounds are described in Section~\ref{sec:fourier}, and we then
use this in Section~\ref{sec:complete_sparsity} to prove the linear rank conjecture 
for complete $d$-uniform polynomials when $d$ is a power of $2$.
In Section~\ref{sec:sparse_sparsity}, we apply our technique to study the sparsity of several more special polynomials 
and prove an upper bound on the granularity of low-degree polynomials.


\section{Preliminaries}\label{sec:prelim}
All logarithms in this paper are base 2. For two $n$-bit vectors $\alpha, \beta \in \Bn$, 
define their inner product as 
$\alpha \cdot \beta = \langle \alpha, \beta \rangle = \sum_{i=1}^n \alpha_i \beta_i \text{ mod }2$
and for simplicity we write $\alpha + \beta$ for $\alpha \oplus \beta$.
We often use $f$ to denote a real function defined on $\Bn$. 
In most occurrences $f$ is a Boolean function, whose range can be represented by either $\B$ or $\pmB$. 
For $f:\BntB$, we use $f^\pm = 1-2f$ to denote the equivalent Boolean function with range converted to $\pmB$. 

\subsection{GF(2) polynomials}
If $S\subseteq [n]$ is a set of (indices of) variables, then the monomial $x_S$ is the product of
variables in $S$: $x_S=\prod_{i \in S}x_i$. The \emph{degree} of this monomial is the cardinality of $S$,
and $S$ is called the \emph{support} of the monomial.
We say a set $T$ \emph{meets} a monomial $x_{S}$ if $T\cap S \neq \emptyset$.

Every Boolean function $f:\BntB$ can be uniquely expressed as a multilinear polynomial over $\gf$:
$p_{f}(x_1,\ldots, x_n)=\sum_{S\subseteq {\mathcal F}}x_{S}$ {where $\mathcal F$ is a collection of subsets of $[n]$} (here additions are performed modulo $2$).
The \emph{degree} of $f$, denoted $\deg_2(f)$, is the maximum degree of its monomials. 
In this paper, whenever there is no risk of confusion, we use $f$ and multilinear polynomial representation of $p_f$ interchangeably.

\subsection{Fourier analysis}
For any real function $f:\Bn\to\mbR$, 
the Fourier coefficients are defined by $\hat{f}(\alpha) = 2^{-n}\sum_{x} f(x)\chi_{\alpha}(x)$, 
where $\chi_{\alpha}(x) = (-1)^{\alpha \cdot x}$. 
The function $f$ can be written as 
$f(x) = \sum_{\alpha} \hat{f}(\alpha) \chi_{\alpha}(x)$. 
The Fourier sparsity of $f$, denoted by $\|\hat f\|_0$, 
is the number of nonzero Fourier coefficients of $f$. 
The Fourier coefficients of $f:\BntB$ and $f^\pm$ are related by 
$\widehat{f^\pm}(\alpha) = \delta_{\alpha,0^n} - 2 \hat f(\alpha)$, 
where $\delta_{x,y}$ is the Kronecker delta function.
Therefore we have 
\begin{equation}\label{eq:range switch}
\fn{f}{0} - 1 \leq \wfn{f^\pm}{0} \leq \fn{f}{0} + 1. 
\end{equation}
Sometimes we employ the one-to-one mapping between vectors in $\cube{n}$ and subsets of $[n]$: $x\leftrightarrow \{i\in [n]: x_i=1\}$,
and use the subsets of $[n]$ to index the Fourier coefficients.

For any function $f:\Bn\to\mbR$, 
Parseval's Identity says that $\sum_{\alpha} \hat{f}^{2}(\alpha) = \av_{x}[f(x)^2]$. 
When the range of $f$ is $\B$, then $\sum_{\alpha} \hat{f}^{2}(\alpha) = \av_x [f(x)]$. 
We sometimes use $\hat f$ to denote the vector of $\{\hat f(\alpha): \alpha\in \Bn\}$. 

\subsection{Granularity and sparsity of Fourier spectrum}
\begin{definition}[Granularity~\cite{GOS+11}]
A rational number $r$ is said to have \emph{granularity} $k$, denoted $\gran(r)=k$, if 
$r=\frac{m}{2^{k}}$ for some {odd} integer $m$. 
The \emph{Fourier granularity} of a Boolean function $f$, denoted $\gran(f)$,
is the maximum granularity over all the Fourier coefficients of $f$;
i.e., $\gran(f)=\max_{\alpha\in \Bn}(\gran(\hat{f}(\alpha)))$.
\end{definition}

Clearly, $\gran(-x)=\gran(x)$ for any $x \in \Q$.
An easy but useful fact is that $\gran(x+y) \leq \max(\gran(x), \gran(y))$
for all $x,y \in \Q$.
More generally, $\gran(\sum_{i=1}^{k}x_{i}) \leq \max_{1\leq i\leq k}\gran(x_{i})$,
where $x_{i}\in \Q$ for every $1\leq i \leq k$.

\begin{fact}\label{fact:granularity_XOR}
Let $f^{\pm},g^{\pm}:\cube{n} \to \{-1,1\}$ be two Boolean functions. 
Let Let $h=f\oplus g$.
Then $|\gran(f^{\pm})-\gran(g^{\pm})| \leq \gran(h^{\pm}) \leq \gran(f^{\pm})+\gran(g^{\pm})$. 
\end{fact}
\begin{proof}
Since the Fourier spectrum of $h^{\pm}$ is given by the convolution formula
\[
\widehat{h^{\pm}}(\alpha) =\sum_{\beta \in \cube{n}}\widehat{f^{\pm}}(\beta)\widehat{g^{\pm}}(\alpha+\beta),
\]
the upper bound on $\gran(h^{\pm})$ follows directly from the definition of granularity. 
Now suppose $\gran(f^{\pm})\geq \gran(g^{\pm})$, then
applying the granularity upper bound on XOR of two functions we just show on $g\oplus h$, which is $f$, gives the desired lower bound. 
\end{proof}

Gopalan \etal~\cite{GOS+11} showed that, if a Boolean function has only a small number of non-zero
Fourier coefficients, then all these non-zero Fourier coefficients have small granularities.

\begin{lemma}[\cite{GOS+11}]\label{lem:gap}
Suppose $f^{\pm}:\cube{n} \to \{-1,1\}$ is $s$-sparse with $s>0$, then all the Fourier coefficients of $f^{\pm}$ 
have granularity at most $\lfloor \log{s}\rfloor - 1$.
\end{lemma}

The following claim shows that the logarithm of the sparsity and granularity of a Boolean function 
are in fact equivalent up to a constant factor. 
\begin{proposition}\label{prop:sparsity_granularity}
Let $f^{\pm}:\cube{n} \to \{-1,1\}$ be a Boolean function, then
\[
\gran(f^{\pm}) + 1 \leq \log \spar(f^{\pm}) \leq 2\gran(f^{\pm}).
\]
\end{proposition}
\begin{proof}
Suppose that $\gran(f^{\pm})=k$. Then for any $\alpha \in \cube{n}$,
if $\hat{f}^{\pm}(\alpha)\neq 0$, then $|\hat{f}^{\pm}(\alpha)|\geq 1/2^{k}$.
By Parseval's identity, we have $\spar(f^{\pm}) \leq 2^{2k}$,
or $\lceil \log(\spar(f^{\pm})) \rceil \leq 2k$. 
Combining with Lemma~\ref{lem:gap} gives the desired result.
\end{proof}

Note that both bounds in Proposition~\ref{prop:sparsity_granularity} 
are tight: for the first inequality, consider the $n$-variate degree-$n$ polynomial 
$f(x)=x_{1}x_{2}\cdots x_{n}$, which 
satisfies $\spar(f^{\pm})=2^{n}$ and $\gran(f^{\pm})=n-1$;
for the second inequality, consider for any even integer $n$ and the Inner Product function
on $n$ variables $f(x)=x_{1}x_{2}+x_{3}x_{4}+\cdots +x_{n-1}x_{n}$,
then $f^{\pm}$ has sparsity $2^{n}$ and granularity $n/2$.

\subsection{Linear maps and restrictions}
Sometimes we need to rotate the input space: For an \emph{invertible} linear map $L$ on $\Bn$, define $Lf$ by $Lf(x) = (f\circ L) (x) = f(Lx)$. 

For a function $f:\Bn\to\mbR$, define two subfunctions $f_0$ and $f_1$, 
both on $\B^{n-1}$: $f_b(x_2, \ldots, x_n) = f(b,x_2, \ldots, x_n)$. 
It is easy to see that for any 
$\alpha \in\B^{n-1}$, $\hat f_b(\alpha) = \hat f(0\alpha) + (-1)^b \hat f(1\alpha)$, thus 
\begin{equation}\label{eq:subfn norm}
	\fn{f_b}{0} \leq \fn{f}{0} \text{ and }\fn{f_b}{1} \leq \fn{f}{1}.
\end{equation}
{where $\|\hat f\|_p = (\sum_{\alpha} |\hat f(\alpha)|^p)^{1/p}$ and $\|\hat f\|_0 = |\{\alpha: \hat f(\alpha)\ne 0\}|$}. 
The notion of subfunctions can be generalized to restrictions with respect to a general direction. 
Suppose $f:\Bn\to\mbR$ and $S\subseteq \Bn$ is a subset of the domain. Then the restriction of $f$ on $S$, denoted by $f|_S$ is the function 
from $S$ to $\mbR$ defined naturally by $f|_S(x) = f(x)$, $\forall x\in S$. 
In this paper, we are concerned with restrictions on affine subspaces.

\begin{lemma}\label{lem:rotation}
Let $f:\Bn\to\mbR$ and $H = a+V$ be an affine subspace, 
then one can (recursively) define the spectrum $\widehat{f|_H}$ of the restricted function $f|_H$ such that
\begin{enumerate}
\item If $\codim(H) = 1$, then $\widehat{f|_H}$ is the collection of 
  $\hat f(\alpha) + (-1)^{b} \hat f(\alpha+\beta)$ for all unordered pair $(\alpha,\alpha+\beta)$, 
  where $\beta$ is the unique non-zero vector orthogonal to $V$, and $b = 0$ if $a \in V$ and $b = 1$ otherwise. 
\item $\wfn{f|_H}{p} \leq \fn{f}{p}$, for any $p\in [0,1]$. 
In particular, restriction does not increase the Fourier sparsity of a function.
\end{enumerate}
\end{lemma}

It is worth noticing that, for any Boolean function, its $\gf$-degree, Fourier sparsity and granularity are all invariant under
invertible linear maps. 

\begin{fact} 
	\label{fact:invertible1}
Let $f$ be an $\gf$-polynomial. Then for any invertible linear map $L$,
$\deg_{2}(f)=\deg_{2}(f \circ L)$.
\end{fact}

\begin{fact} 
	\label{fact:invertible2}
Let $f:\BntB$ be a Boolean function and $L$ an invertible linear map. Then {$\widehat{f \circ L}(\alpha) = \hat{f}((L^T)^{-1} \alpha)$}. In particular, 
$\spar(f)=\spar(f \circ L)$
and $\gran(f)=\gran(f \circ L)$.
\end{fact}


\section{Linear rank of complete $d$-uniform maxonomials}\label{sec:lrank}

We now compute the exact value of the linear rank of a degree $d$ polynomial whose set of maxonomials
consists of all $\binom{n}{d}$ degree-$d$ monomials, and give explicit linear constraints restriction of which reduces
the degree of such a polynomial.

Define $\mathcal{C}_{d,n}(x) = \sum_{I \subseteq [n], |I| = d} \prod_{i\in I}x_i$, the summation of all degree-$d$ monomials over variables $x_1, \ldots, x_n\in \gf$. The subscript $n$ is dropped when it is clear from the context. We use the equivalence relation $\eqd$ for polynomials with the same maxonomials, \ie $p \eqd q$ if both $p$ and $q$ have \gf-degree $d$ and $p+q$ 
has \gf-degree strictly less than $d$. It is clear that if $p \eqd q$, then $\lrank(p) = \lrank(q)$.

\begin{theorem}
\label{thm:symlrank}
Let $n\ge d\ge 0$ be integers. Then the following hold:
\begin{enumerate}
\item
If $d$ is odd, then $\lrank(\mathcal{C}_{d,n}) = 1$.
\item
If  $d$ is even, then $\lrank(\mathcal{C}_{d,n}) = \floor{\frac{n}{2}}-\frac{d}{2}+1$, i.e.
\begin{equation*}
\lrank(\mathcal{C}_{d,n}) =
\begin{cases}
\frac{n-d}{2}+1 & \text{if $n$ is even},\\
\frac{n-d-1}{2}+1 & \text{if $n$ is odd}.
\end{cases}
\end{equation*}
\end{enumerate}
\end{theorem}

\begin{proof}
The first item follows simply by the factorization $\mathcal{C}_{d,n} \eqd \mathcal{C}_{1,n}\mathcal{C}_{d-1,n}$. 
Indeed, when we multiply $\mathcal{C}_{1,n} = \sum_{i\in [n]} x_i$ and 
$\mathcal{C}_{d-1,n} = \sum_{|I| = d-1} x_I$, for $i\notin I$, $x_i x_I = x_{I\cup\{i\}}$, 
and each $J$ with $|J| = d$ comes from $d$ many $(i,I)$. 
For each $i\in I$, $x_i x_I = x_I$, and each resulting $x_I$ with $|I| = d-1$ comes from $d-1$ many $i \in I$. Thus 
\begin{align*}
\mathcal{C}_{1,n} \mathcal{C}_{d-1,n} & = d \Big(\sum_{|J| = d} x_J \Big) + 
(d-1)\Big(\sum_{|I| = d-1} x_I \Big) = d\mathcal{C}_{d,n}+(d-1)\mathcal{C}_{d-1,n} \\
& = \mathcal{C}_{d,n},
\end{align*}
for all odd $d$.

Now we consider the second item in the statement and assume from now on that $d$ is even and $d \leq n$. 
The second item follows from the following two claims. 

\begin{claim}
\label{claim:lowerb}
If $\lrank(\mathcal{C}_{d,n+1}) = \lrank(\mathcal{C}_{d,n})$, then $\lrank(\mathcal{C}_{d,n+2}) > \lrank(\mathcal{C}_{d,n+1})$.
\end{claim}

\begin{claim}
\label{claim:upperb}
$\lrank(\mathcal{C}_{d,n+2}) \leq \lrank(\mathcal{C}_{d,n})+1$.
\end{claim}

Let us first show Theorem \ref{thm:symlrank} assuming these two lemmas. We prove by induction on the number of variables that for all $k \geq d/2$,
\begin{equation} \label{eqn:induction}
\lrank(\mathcal{C}_{d,2k}) = \lrank(\mathcal{C}_{d,2k+1}) = k-\frac{d}{2}+1.
\end{equation}
which is just a restatement of the second item of Theorem \ref{thm:symlrank}. 
\paragraph{base case $k = d/2$.} We have
\begin{align}
\mathcal{C}_{d}(x_1, \ldots, x_{2k}) = \mathcal{C}_{d}(x_1, \ldots, x_{d}) =  \mathcal{C}_{d-1}(x_1, \ldots, x_{d-1})\cdot x_{d}, \label{eqn:even}
\end{align}
so $\lrank(\mathcal{C}_{d,2k}) = 1$. For $n=2k+1$, note that
\begin{align}
\mathcal{C}_{d}(x_1, \ldots, x_{2k+1}) 
& = \mathcal{C}_{d}(x_1, \ldots, x_{d+1}) \nonumber \\
& = \mathcal{C}_{d-1}(x_1, \ldots, x_{d-1})(x_{d}+x_{d+1})+\mathcal{C}_{d-2}(x_1, \ldots, x_{d-1})x_{d}x_{d+1}, \label{eqn:odd}
\end{align} 
Putting restriction $x_d=x_{d+1}$ makes the first summand vanish and decreases the degree of the second summand, hence $\lrank(\mathcal{C}_{d,2k+1}) = 1$. 

\paragraph{general $k$.} Now we assume that Eq.~\eqref{eqn:induction} holds for $k$ and will prove the case for $k+1$. The following sequence of inequalities hold.
\begin{equation*}
k-\frac{d}{2}+1 < \lrank(\mathcal{C}_{d,2(k+1)}) \leq \lrank(\mathcal{C}_{d,2(k+1)+1}) \leq k-\frac{d}{2}+2,
\end{equation*}
where the first inequality follows by Claim~\ref{claim:lowerb}; 
the second follows by the facts that $\mathcal{C}_{d,n-1}$ can be obtained from $\mathcal{C}_{d,n}$ by restricting $x_n = 0$ 
and restriction does not increase $\lrank$; 
and the last inequality follows by Claim~\ref{claim:upperb}. 
Therefore Eq.~\eqref{eqn:induction} also holds for $k+1$.
\end{proof}

Now it remains to prove the two claims. We start with Claim~\ref{claim:upperb}, which is simpler.

\begin{proof}[Proof of Claim~\ref{claim:upperb}]
We first observe the following identity:
\begin{align}
\mathcal{C}_{d}(x_1, \dots, x_{n+2}) 
&= \mathcal{C}_{d}(x_1, \dots, x_n) + \mathcal{C}_{d-1}(x_1, \dots, x_n)(x_{n+1} + x_{n+2}) 
   + \mathcal{C}_{d-2}(x_1, \dots, x_n) x_{n+1} x_{n+2} \nonumber \\
&\eqd \mathcal{C}_{d}(x_1, \dots, x_n) + \mathcal{C}_{d-1}(x_1, \dots, x_n, x_{n+1})(x_{n+1} + x_{n+2}). \label{eqn:recurrence} 
\end{align}
Therefore the restriction $x_{n+2} = x_{n+1}$ reduces $\mathcal{C}_{d}(x_1, \ldots, x_{n+2})$ to
\begin{align*}
\mathcal{C}_{d}(x_1, \ldots, x_{n+2})|_{x_{n+1}=x_{n+2}} \eqd \mathcal{C}_{d}(x_1, \dots, x_n).
\end{align*}

Since each restriction can reduce $\lrank$ by at most 1, we have
\begin{equation*}
\lrank(\mathcal{C}_{d,n+2})-1 \leq \lrank(\mathcal{C}_{d,n+2}|_{x_{n+2} = x_{n+1}}) = \lrank(\mathcal{C}_{d,n}),
\end{equation*}
as desired.
\end{proof}

\begin{proof}[Proof of Claim~\ref{claim:lowerb}]
For the sake of contradiction, assume that
\begin{equation*}
\lrank(\mathcal{C}_{d,n+2}) = \lrank(\mathcal{C}_{d,n+1}) = \lrank(\mathcal{C}_{d,n}) = r.
\end{equation*}
Fix an optimal set of linear restrictions for $\lrank(\mathcal{C}_{d,n+2})$. Without loss of generality, 
we can assume it contains a restriction of the form $x_{n+2} = \ell(x_1, \dots, x_{n+1}) = \ell(x)$ for some linear form $\ell$. 
It is clear that such restriction will reduce the $\lrank$ by exactly 1. So we have
\begin{equation}\label{eq:lrC2}
\lrank(\mathcal{C}_{d,n+2}|_{x_{n+2} = \ell(x)}) \leq \lrank(\mathcal{C}_{d,n+2})-1 = r-1.
\end{equation}
But by the expansion
\begin{equation*}
\mathcal{C}_{d}(x_1, \dots, x_{m+1}) = \mathcal{C}_d(x_1, \dots, x_m) + \mathcal{C}_{d-1}(x_1, \dots, x_m)x_{m+1},
\end{equation*}
we have
\begin{align}
\mathcal{C}_{d}(x_1, \dots, x_{n+2})|_{x_{n+2} = \ell(x)} & = \mathcal{C}_d(x_1, \dots, x_{n+1}) + \mathcal{C}_{d-1}(x_1, \dots, x_{n+1})\ell(x) \nonumber \\
& =  \mathcal{C}_d(x_1, \dots, x_n) + \mathcal{C}_{d-1}(x_1, \dots, x_n)x_{n+1} + \mathcal{C}_{d-1}(x_1, \dots, x_{n+1})\ell(x).
\label{eq:lrC3}
\end{align}
Now, consider to further restrict $x_{n+1} = x_1 + x_2 + \dots + x_n = \mathcal{C}_{1}(x_1, \dots, x_n)$. 
By the fact that $\mathcal{C}_{d-1}(x_1, \dots, x_m) \eqd \mathcal{C}_{d-2}(x_1, \dots, x_m)\mathcal{C}_{1}(x_1, \dots, x_m)$ 
for every even $d \geq 4$, 
the second term on the right of Eq.\eqref{eq:lrC3} is $\eqd$-equivalent to
\begin{align*}
& \mathcal{C}_{d-2}(x_1, \dots, x_n)\mathcal{C}_1(x_1, \dots, x_n)x_{n+1}|_{x_{n+1} = \mathcal{C}_1(x_1, \dots, x_n)} \\
= \ & \mathcal{C}_{d-2}(x_1, \dots, x_n)\mathcal{C}_1^2(x_1, \dots, x_n) \\
= \ & \mathcal{C}_{d-2}(x_1, \dots, x_n)\mathcal{C}_1(x_1, \dots, x_n) \eqd 0,
\end{align*}
and the last term becomes
\begin{equation*}
\mathcal{C}_{d-2}(x_1, \dots, x_{n+1})\mathcal{C}_1(x_1, \dots, x_{n+1}) \ell(x)|_{x_{n+1} = \mathcal{C}_1(x_1, \dots, x_n)} = 0.
\end{equation*}
Plugging these two back to Eq.\eqref{eq:lrC3}, 
\begin{equation*}
\mathcal{C}_{d,n+2}|_{x_{n+2} = \ell(x), x_{n+1} = x_1 + \dots + x_{n}} \eqd \mathcal{C}_{d,n}.
\end{equation*}
As restriction does not increase linear rank, we have from Eq.\eqref{eq:lrC2} that
\begin{equation*}
r = \lrank(\mathcal{C}_{d,n}) = \lrank(\mathcal{C}_{d,n+2}|_{x_{n+2} = \ell(x), x_{n+1} = x_1 + \dots + x_{n}}) \le \lrank(\mathcal{C}_{d,n+2}|_{x_{n+2} = \ell(x)}) \leq r-1,
\end{equation*}
which is a contradiction.
\end{proof}

As a simple application of Theorem~\ref{thm:symlrank}, for any symmetric function $f$, let $r_1$, $r_0$ be the largest and smallest integers such that $f(x)$ is constant or parity on $\{x\in \cube{n} : r_0 \leq |x| \leq n-r_1\}$. The quantity $r \defeq r_0 + r_1$ turns out to be an important complexity measure for symmetric functions. For example, the randomized and quantum communication complexity of symmetric XOR functions is characterized by this $r$ (\cite{ZS09,LLZ11,LZ13}), and $\log\fn{f}{1} = \Theta(r\log(n/r))$ for all symmetric functions $f$ (\cite{AFH12}). 

Here we relate this measure to the \gf-degree of $f$. It is clear that we can fix $x_1 = x_2 = \dots = x_{r_0} = 1$ and $x_n = x_{n-1} = \dots = x_{n - r_1 + 1} = 0$ to reduce the degree of $f$ to at most $1$. We therefore have the following corollary.

\begin{corollary}
Let $f$ be a symmetric function with even \gf-degree $d$, then 
\begin{enumerate}
	\item $\floor{\frac{n}{2}}-\frac{d}{2}+1 \leq r_0 + r_1$.
	\item $\log \fn{f}{1} = \Omega(n/\log n)$, if $d = (1-\Omega(1))n$.
\end{enumerate} 
\end{corollary}

\ignore{
One can also obtain a lower bound on $\|\hat{f}\|_1$ in terms of $n, d$ using the characterization of 
spectral norm in terms of $r_0$ and $r_1$. \red{to do: easy but need to check details out.}
}

\subsection{An explicit form of linear restrictions for complete $d$-uniform monomials}
The proof of Theorem~\ref{thm:symlrank} can be used to find a linear transformation
which explicitly show the restrictions for $\mathcal{C}_{d,n}$. Indeed, 
starting from either Eq.~\eqref{eqn:even} or Eq.~\eqref{eqn:odd} and recursively applying Eq.~\eqref{eqn:recurrence},
gives, when $n=d+2k$ is even,
\begin{align*}
&\quad \mathcal{C}_{d}(x_1, \ldots, x_{n}) \\
&\eqd \mathcal{C}_{d}(x_1, \ldots, x_{n-2})+(x_{n-1}+x_{n})\mathcal{C}_{d-1}(x_1, \ldots, x_{n-1}) \\
&\eqd \mathcal{C}_{d}(x_1, \ldots, x_{n-4})+ (x_{n-3}+x_{n-2})\mathcal{C}_{d-1}(x_1, \ldots, x_{n-3}) + 
   (x_{n-1}+x_{n})\mathcal{C}_{d-1}(x_1, \ldots, x_{n-1}) \\
&\eqd \cdots \cdots \\
&\eqd \mathcal{C}_{d}(x_1, \ldots, x_{d})+(x_{d+1}+x_{d+2})\mathcal{C}_{d-1}(x_1, \ldots, x_{d+1}) + \cdots + 
(x_{n-1}+x_{n})\mathcal{C}_{d-1}(x_1, \ldots, x_{n-1})\\ 
&= x_{d}\mathcal{C}_{d-1}(x_1, \ldots, x_{d-1})+(x_{d+1}+x_{d+2})\mathcal{C}_{d-1}(x_1, \ldots, x_{d+1}) + \cdots + 
(x_{n-1}+x_{n})\mathcal{C}_{d-1}(x_1, \ldots, x_{n-1}).
\end{align*}
Then in the new basis where $y_1=x_1, \ldots, y_d=x_d, y_{d+1}=x_{d+1}, y_{d+2}=x_{d+1}+x_{d+2}, \ldots,
y_{n-1}=x_{n-1}, y_{n}=x_{n-1}+x_{n}$, we have
\begin{align*}
& \mathcal{C}_{d}(x_1, \ldots, x_{n}) 
   = \mathcal{C}_{d}(y_1, \ldots, y_{d}, y_{d+1}, y_{d+1}+y_{d+2}, \ldots, y_{n-1}, y_{n-1}+y_{n}) \\
\eqd& ~ y_{d}\mathcal{C}_{d-1}(y_1, \ldots, y_{d-1})+y_{d+2}\mathcal{C}_{d-1}(y_1, \ldots, y_{d}, y_{d+1}) + 
 y_{d+4}\mathcal{C}_{d-1}(y_1, \ldots, y_{d}, y_{d+1}, y_{d+1}+y_{d+2}, y_{d+3})  \\
&\quad + \cdots + 
  y_{n}\mathcal{C}_{d-1}(y_1, \ldots, y_{d}, y_{d+1}, y_{d+1}+y_{d+2}, y_{d+3}, y_{d+3}+y_{d+4}, \ldots, y_{n-3}, y_{n-3}+y_{n-2}, y_{n-1}).
\end{align*}
Hence $\{y_{d}, y_{d+2}, \ldots, y_{n}\}$ is a set of $k+1=\floor{\frac{n}{2}}-\frac{d}{2}+1$ linear
restrictions that reduce $\mathcal{C}_{d,n}$'s degree. By Theorem~\ref{thm:symlrank}, this is the best possible.

Similarly, when $n=d+2k+1$ is odd,
\begin{align*}
& \quad \mathcal{C}_{d}(x_1, \ldots, x_{n}) \\
&\eqd \mathcal{C}_{d}(x_1, \ldots, x_{n-2})+(x_{n-1}+x_{n})\mathcal{C}_{d-1}(x_1, \ldots, x_{n-1}) \\
&\eqd \cdots \cdots \\
&\eqd \mathcal{C}_{d}(x_1, \ldots, x_{d+1})+(x_{d+2}+x_{d+3})\mathcal{C}_{d-1}(x_1, \ldots, x_{d+2}) + \cdots + 
(x_{n-1}+x_{n})\mathcal{C}_{d-1}(x_1, \ldots, x_{n-1}) \\
&\eqd (x_{d}+x_{d+1})\mathcal{C}_{d-1}(x_1, \ldots, x_{d}) + \cdots + (x_{n-1}+x_{n})\mathcal{C}_{d-1}(x_1, \ldots, x_{n-1}). 
\end{align*}

Now if we switch to the basis in which $y_1=x_1, \ldots, y_d=x_d, y_{d+1}=x_{d}+x_{d+1}, \ldots,
y_{n-1}=x_{n-1}, y_{n}=x_{n-1}+x_{n}$, then
\begin{align*}
& \mathcal{C}_{d}(x_1, \ldots, x_{n}) 
   = \mathcal{C}_{d}(y_1, \ldots, y_{d}, y_{d}+y_{d+1}, \ldots, y_{n-1}, y_{n-1}+y_{n}) \\
\eqd& ~ y_{d+1}\mathcal{C}_{d-1}(y_1, \ldots, y_{d})+y_{d+3}\mathcal{C}_{d-1}(y_1, \ldots, y_{d}, y_{d}+y_{d+1}, y_{d+2}) + \cdots + \\
&\quad y_{n}\mathcal{C}_{d-1}(y_1, \ldots, y_{d}, y_{d}+y_{d+1}, \ldots, y_{n-3}, y_{n-3}+y_{n-2}, y_{n-1}).
\end{align*}
Consequently, $\{y_{d+1}, y_{d+3}, \ldots, y_{n}\}$ is a set of $k+1=\floor{\frac{n}{2}}-\frac{d}{2}+1$ linear
restrictions that reduce $\mathcal{C}_{d,n}$'s degree and meet the bound in Theorem~\ref{thm:symlrank}.


\section{Fourier spectra of GF(2) polynomials} \label{sec:fourier}

In this Section, we present a framework for computing the Fourier spectrum of a GF(2) polynomial based on its monomials.
We suspect that such a formalism was known before but we could not track any previous sources.

For a fixed $S\subseteq [n]$, a collection $\{S_1, \ldots, S_k\}$ of $k$ (distinct) subsets of $[n]$ 
form a \emph{$k$-cover} of $S$ if $\cup_{i=1}^{k}S_{i}=S$.
The main result of this section is the following lemma, which shows that the Fourier coefficients of
a GF(2) polynomial can be computed by counting the number of $k$-covers of subsets of $[n]$ 
--- for different values of $k$ ---
using the supports of monomials in the GF(2) polynomial as subsets. 
Of particular importance is the number of $k_\text{min}$-covers of $[n]$,
where $k_\text{min}$ is minimum number of subsets that are required to cover $[n]$.


For a family $\mathcal{F}=\{S_i\}_{i\in [m]}$ of subsets $S_i$ of the base set $[n]$ and an index set $M\subseteq [m]$, let $S_{M}\eqdef \cup_{k\in M} S_k$, the union of the subsets with indices in $M$. 



Let $f(x_1, \ldots, x_n) = \sum_{i=1}^m x_{S_i}$ be the GF(2) polynomial representation of $f$. Define a \emph{weight} function $w_f:\cube{n}\to \Q$ as
\begin{align}\label{eq:weight1}
w_f(T)=\sum_{\text{$M\subseteq [m]$: $S_M=T$}}c(M), \quad \text{ where } \ c(M)=\frac{(-2)^{|M|}}{2^{|S_M|}}.
\end{align}
Equivalently, if we denote $\mathcal{F}=\{S_i\}_{i\in [m]}$ and let $N_{k}(T)$ be the number of $k$-covers of $T$ using sets in $\mathcal{F}$, then
\begin{align}\label{eq:weight2}
w_f(T) = \frac{1}{2^{|T|}}\sum_{k=1}^{m}(-2)^{k}N_{k}(T).
\end{align} 
\begin{lemma} 
\label{lem:main}
Let $f(x_1, \ldots, x_n) = \sum_{i=1}^m x_{S_i}$ be a GF(2) polynomial, then the Fourier coefficients of $f^{\pm}$ are given by
\begin{align}\label{eq:Fourier_main}
\widehat{f^{\pm}}(S) = (-1)^{|S|} \sum_{T\supseteq S} w_f(T).
\end{align}
\end{lemma}

\begin{proof}
For a Boolean variable $x_i \in \{0,1\}$, let $\tilde{x_i} = (-1)^{x_i} = 1-2x_i$ be its $\pmB$ representation, with the inverse transformation given by $x_i = (1-\tilde{x_i})/2$. Recall that $f^{\pm}=1-2f$. We next express $f^{\pm}$ as a multilinear polynomial over $\mbR$ from which its Fourier coefficients can be readily read out.


Note that $x_S$ corresponds to $1-2\prod_{i\in S}\frac{1-\tilde x_i}{2}$ and 
$\prod_{i\in S} \tilde{x}_{i}$ corresponds to $\tilde x_S$, thus 
\begin{align} \label{eqn:Fourier_expansion}
f^{\pm}(\tilde{x}_{1}, \ldots, \tilde{x}_{n}) = \prod_{i\in [m]} \Big(1-2\prod_{j\in S_i} \frac{1-\tilde x_{j}}{2}\Big)
\end{align}

\begin{fact}\label{fact:powering}
For $x\in \{-1, 1\}$ and integer $k\geq 1$, we have $(1-x)^{k}=2^{k-1}(1-x)$.
\end{fact}

By Eq.\eqref{eqn:Fourier_expansion}, the Fourier polynomial of $f^{\pm}$ 
in terms of $\tilde{x}$ is
\begin{align*}
  f^{\pm}(\tilde{x}) 
&=\prod_{i=1}^m \left(1-\frac{\prod_{j\in S_i} (1-\tilde{x}_{j})}{2^{|S_i|-1}} \right) \\
&=\sum_{k=0}^m (-1)^k \sum_{1\leq i_{1} < i_{2}< \ldots < i_k \leq m}
  \frac{\prod_{j_{1} \in S_{i_1}}(1-\tilde{x}_{j_1}) \prod_{j_{2} \in S_{i_2}}(1-\tilde{x}_{j_2})
  \cdots \prod_{j_{k} \in S_{i_k}}(1-\tilde{x}_{j_k})}{2^{|S_{i_1}|+|S_{i_2}|+\cdots+|S_{i_k}|-k}} \\ 
&=\sum_{k=0}^m (-1)^k \sum_{1\leq i_{1} < i_{2}< \ldots < i_k \leq m} 
   \frac{\prod_{j\in S_{i_1} \cup \cdots \cup S_{i_k}}(1-\tilde{x}_{j})}{2^{|S_{i_1} \cup \cdots 
     \cup S_{i_k}|-k}}	 \qquad (\text{by Fact~\ref{fact:powering}}) \\
&=\sum_{M\subseteq [m]} (-1)^{|M|} \frac{\prod_{j\in S_{M}}(1-\tilde{x}_{j})}{2^{|S_{M}|-|M|}} \\
&=\sum_{S\subseteq [n]} (-1)^{|S|}\left(\sum_{M\subseteq [m]: S_{M}\supseteq S} 
  (-1)^{|M|}\cdot \frac{2^{|M|}}{2^{|S_{M}|}}\right) \tilde{x}_{S},
\end{align*}
Since the coefficient of  $\tilde{x}_{S}$ in 
 $f^{\pm}(\tilde{x})$ is just the Fourier coefficient $\widehat{f^{\pm}}(S)$, this completes the proof of the lemma.
\end{proof}

The  weight function's value at $[n]$, $w_f([n])$, is the a particularly important term, as it contributes to all the Fourier coefficients of $f^{\pm}$. In particular, if the granularity of $w_f([n])$
is larger than the granularity of any other $w_f(T)$, then all Fourier coefficients of $f^{\pm}$ are non-zero. This will be used to lower bound granularity for different functions in the next two sections.

\section{Fourier sparsity of polynomials with complete $d$-uniform maxonomials} \label{sec:complete_sparsity}

This section is devoted to the proof of the following Fourier sparsity lower bound for polynomials whose
maxonomials are the complete $d$-uniform monomials.
\begin{theorem}\label{thm:symmetric_sparsity}
Let $d$ be a power of $2$. For any degree-$d$ polynomial $f\in \gf[x_1,\ldots,x_n]$ whose maxonomials include all $\binom{n}{d}$ degree-$d$ monomials, its Fourier sparsity has the following lower bound 
\[\spar(f)\geq 2^{d\cdot\lfloor n/d\rfloor}-1=\Omega(2^{n}),\]
regardless of the lower degree monomials. 
\end{theorem}

\begin{remark}
In the rest of this section, we fix $k=\lfloor n/d \rfloor$. 
\end{remark}

First we apply a restriction to set, say the last $n-kd$ variables in $f$ to zero.
This leaves us with a function $g$ on $n'=kd$ variables, and by Item 2 of Lemma~\ref{lem:rotation}, 
$\spar(f)\geq \spar(g)$. 
Furthermore, the maxonomials of $g$ are still complete $d$-uniform monomials (now over $n'$ variables). 

Let $\mathcal{F}$ be the set of the supports of all monomials in $g$. In particular,
$\mathcal{F}$ contains all $d$-subsets of $[n']$: $\binom{[n']}{d}\subseteq \mathcal{F}$.
 
\begin{lemma}\label{lem:complete_symmetric_weight1}
The granularity of the  weight function at $[n']$ (hence the Fourier coefficient of $g$ at $[n']$) 
is $\gran\left(w_g([n'])\right)= n'-k$.
\end{lemma}
\begin{proof}
By Lemma~\ref{lem:main}, as the degrees of all monomials in $g$ are at most $d$, 
the minimum number of subsets required from $\mathcal{F}$ to cover $[n']$ is $k$, therefore 
\begin{align}\label{eq:symmetric_weight}
\widehat{g^{\pm}}([n'])=(-1)^{n'}w_g([n'])=(-1)^{n'}\sum_{j=k}^{m}\frac{(-2)^{j} N_{j}([n'])}{2^{n'}}.
\end{align}
\begin{claim}
$ N_{k}([n']) \equiv 1 \pmod{2}$.
\end{claim}
\begin{proof}
Clearly any $k$-cover of $[n']$ consists of $k$ distinct sets in $\binom{[n']}{d}$, 
and there are exactly $\frac{\binom{n'}{d,\ldots, d}}{k!}$
such $k$-covers. Hence we have
\begin{align*}
 N_{k}([n']) &= \frac{\binom{n'}{d,\ldots, d}}{k!} = \frac{1}{k}\binom{kd}{d} \cdot \frac{1}{k-1}\binom{(k-1)d}{d} \cdots 1\cdot \binom{d}{d} \\
 &=\binom{kd-1}{d-1} \cdot \binom{(k-1)d-1}{d-1} \cdots \binom{d-1}{d-1}.
\end{align*}

Recall the following Lucas' theorem: 
\begin{theorem}[Lucas' theorem, c.f.~\cite{Fin47}]
Let $s$ and $t$ be non-negative integers and $p$ be a prime. 
Let $s=s_0+s_1 p +\cdots s_i p^i$ and $t=t_0+t_1 p +\cdots t_i p^i$, $0\leq s_j, t_j <p$,
be the base-$p$ expansions of $s$ and $t$ respectively,
then
\[
\binom{s}{t} \equiv \prod_{j=0}^{i}\binom{s_j}{t_j} \pmod{p}.
\]
\end{theorem}
In fact, what we need is the a simple corollary of Lucas' theorem (known as Kummer's theorem) for the special case of $p=2$: 
the largest integer $j$ such that $2^j$ divides $\binom{s}{t}$
is equal to the number of carries that occur when $s$ and $s-t$ are added in the binary.
 
Since $d$ is a power of $2$, the binary representation of $d-1$ is $\underbrace{1\cdots 1}_{\log{d}}$ and
the binary representation of $jd-1-(d-1)=(j-1)d$ is $\cdots \underbrace{0\cdots 0}_{\log{d}}$, for every $j\geq 1$.
Therefore no carry occurs when adding $(j-1)d$ to $d-1$ and 
thus, by Kummer's theorem, $\binom{jd-1}{d-1}\equiv 1\pmod{2}$ for all $j\geq 1$. 
It follows that $ N_{k}([n']) \equiv 1 \pmod{2}$.
\end{proof}

Finally note that the granularity of the  $(j-k+1)^{\text{st}}$ term in Eq.~\eqref{eq:symmetric_weight} satisfies
\[
\gran\left(\frac{(-2)^{j} N_{j}([n'])}{2^{n'}}\right) \leq n'-j < n'-k,
\]
for all $j>k$, therefore the first term is the unique term in the sum which has the highest granularity $n'-k$.
Hence its granularity is also the granularity of the sum in Eq.~\eqref{eq:symmetric_weight}.
This completes the proof of Lemma~\ref{lem:complete_symmetric_weight1}.
\end{proof}

Now we need the following simple observations, which are simple consequences of Lemma~\ref{lem:main}.
\begin{fact}\label{fact:max_gran_w}
Let $g:\cube{n'}\to \B$ be a degree-$d$ polynomial. Then for any $T\subseteq [n']$,
the granularity of the  weight function of $g$ at $T$ $\gran(w_{g}(T))$ is at most $|T|-\lceil |T|/d \rceil$. 
\end{fact}
\begin{proof}
This follows directly from Eq.~\eqref{eq:weight2}: since every subset in $\mathcal{F}$ is of size at most $d$,
the minimum number of sets to cover $T$ is $\lceil |T|/d \rceil$.
\end{proof}

As a simple corollary of Fact~\ref{fact:max_gran_w}, we have
\begin{corollary}\label{cor:complete_symmetric_weight2}
Let $g:\cube{n'}\to \B$ be a degree $d$ polynomial.
Then for any $T\subseteq [n']$, 
$\gran(w_{g}(T))\leq n'-\lceil n'/d \rceil$,
and equality is only possible for $T=[n']$.
\end{corollary}
In other words, if the granularity of $w_{g}([n'])$ is indeed equal to $n'-\lceil n'/d\rceil$, then
that is the unique highest granularity among all  weight values.

Now applying Proposition \ref{prop:sparsity_granularity} gives $\spar(g) \ge 2^{n'-k}$. To get the stronger lower bound $2^{n'}$ as claimed, 
let us combine Lemma~\ref{lem:complete_symmetric_weight1}, Corollary~\ref{cor:complete_symmetric_weight2} 
and Eq.~\eqref{eq:weight2} in Lemma~\ref{lem:main}, and observe that not only $w_g([n'])$ has the unique highest granularity among all  
weights $\{w_g(S)\}_{S\subseteq [n']}$,
but also it is included in the Fourier coefficient of $\widehat{g^{\pm}}(S)$ for every $S\subseteq [n']$.
We therefore 
see that for all $S\subseteq [n']$, $\gran\left(\widehat{g^{\pm}}(S)\right)=n'-k>0$; consequently
$\spar\left(g^{\pm}\right)=2^{n'}$. 
It follows that 
\[\spar(g) \ge \spar(g^{\pm}) - 1 = 2^{n'} - 1 = 2^{d\lfloor n/d\rfloor} - 1,\]
completing the proof of Theorem~\ref{thm:symmetric_sparsity}.


\section{Fourier sparsity for functions with sparse maxonomials} \label{sec:sparse_sparsity}
In the previous two sections, we see cases that when all $\binom{n}{d}$ monomials of the highest degree appear, then the function has large Fourier sparsity, no matter what other lower-degree monomials exist or not. 
In this section, we will consider the other end of the spectrum when there are only a small number of the maxonomials, and show that the same phenomena can occur in this case as well. 

The first example is the class of functions with disjoint maxonomials. 

\begin{proposition}
Suppose that $f:\BntB$ has $\deg_2(f) = d$ where $d|n$. If there are exactly $n/d$ monomials of degree $d$, 
and their supports are pairwise disjoint, then $\spar(f) \ge 2^{n}-1$, regardless of the lower degree monomials. 
\end{proposition}
\begin{proof}
We apply Lemma~\ref{lem:main} to $f$ and note that the smallest number of sets needed to cover $[n]$ is $n/d$, 
achieved by the maxonomials. Thus the Fourier coefficient $\widehat{f^\pm}([n])$ equals $ \pm \frac{1}{2^{n-n/d}}$ 
plus some fractions with denominator $2^{k}$ for some $k<n-n/d$. 
Therefore $\gran(\widehat{f^\pm}([n])) = n-n/d$. 
Now using a similar argument as the last part of the proof for Theorem \ref{thm:symmetric_sparsity}, 
we see that all Fourier coefficients of $f^\pm$ are non-zero. Thus $\spar(f) \ge \spar(f^\pm)-1 \ge 2^{n}-1.$ 
\end{proof}

The second example extends the first class by allowing ``regular'' overlaps between maxonomials. Assume that $\deg_2(f) = d$ is an odd prime power, and $d^2|n$. Divide $[n]$ into $n/d^2$ piles of equal size, with each pile identified with a $d\times d$ grid. All maxonomials are linear functions in a pile. More precisely, for the first pile $[d]\times [d]$, for each pair $(a,b) \in \mathbb{F}_d^2$, define univariate polynomial $p_{a,b}\in \mathbb{F_d}[x]$ by $p_{a,b}(x) = ax+b$. Now define sets 
\[S_{a,b} = \{(0,p(0)), (1,p(1)), \ldots, (d-1,p(d-1))\}.\]
The first pile thus has $d^2$ sets inside. Similarly define $d^2$ sets for each other pile. 
These sets are supports of the maxomonials. Note that there are $d^2 \cdot n/d^2 = n$ maxonomials, 
a number much smaller than the possible number of lower degree monomials, which is $\sum_{i=0}^{d-1} \binom{n}{i}$. 
Yet the next theorem says that the this small number of maxonomials determines a large Fourier sparsity, 
regardless of how the vast majority of other (lower-degree) terms behave. 

\begin{theorem}
For any function $f:\BntB$ with the maxonomials defined as above, $\spar(f) \ge 2^{n}-1$, regardless of the lower degree monomials.
\end{theorem}	
\begin{proof}
Clearly the set $[n]$ can be partitioned using supports of $n/d$ maxonomials. 
We will show that the number of such partitions is $d^{n/d^2}$, which is an odd number given that $d$ is odd. 
	
Since the piles are disjoint and all maxonomials are defined within each pile, 
it suffices to show that there are $d$ ways of partitioning each pile into maxonomials. 
We consider the first pile and the same argument applies to others. 
Note that for each fixed $a$, if we vary $b$ over $\mathbb{F}_d$, 
then we get $d$ maxonomials that are pairwise disjoint. 
Since there are $d$ different choices of $a$, there are at least these $d$ ways to partition the pile into $d$ maxonomials. 
We next show that there are actually no other partition of the pile using $d$ maxonomials. 
Indeed, assume that a partition uses $d$ maxonomials and not all these maxonomials have the same $a$, 
then there are two maxonomials corresponding to $a_1 x + b_1$ and $a_2 x + b_2$ and $a_1 \ne a_2$. 
But now these two ``lines'' intersect at exactly one point $x = (a_1 - a_2)^{-1}(b_1 - b_2)$, 
where the existence of $(a_1 - a_2)^{-1}$ uses the assumption that $a_1 \ne a_2$. 
Note the trivial fact that the union of $d$ maxonomials of degree $d$ is at most $d^2$, and it is $d^2$ only if they are pairwise disjoint. 
So the existence of intersecting maxonomials in the selected $d$ maxonomials make them impossible to cover the $d^2$ points in the pile. 
This shows that the number of  partitions of one pile using $d$ maxonomials is exactly $d$, 
and thus the number of covers of $[n]$ using $n/d^2$ maxonomials is $d^{n/d^2}$. 
Now apply a similar argument as the last part of the proof for Theorem \ref{thm:symmetric_sparsity}, 
we see that $\spar(f) \ge \spar(f^\pm)-1 \ge 2^{n}-1$.
\end{proof}	


It would be nice to also pin down the linear rank of the functions with the maxonomials defined as above. What we are able to say at this moment is an upper bound only.
\[\lrank(f) \le n/d.\]
Indeed, for each pile, we can pick the first column of variables and set them all to 0. This makes all maxonomials vanish, and thus decreases the degree by at least 1. 

\subsection{Granularity upper bound for low-degree polynomials}
Note that there is a gap of factor $2$ in characterizing the logarithm of Fourier sparsity of a Boolean function
by means of its granularity (cf. Proposition~\ref{prop:sparsity_granularity}).
Note that both lower and upper bounds in Proposition~\ref{prop:sparsity_granularity} are tight,
but one is attained by the AND function (a degree-$n$ polynomial) and the other by any bent function, e.g. 
the Inner Product function (a degree-$2$ polynomial).
It thus natural to conjecture that, for any low-degree polynomial $f(x)$, 
although $\spar(f)$ can be as large as $2^n$, the granularity of $f(x)$ is always bounded away from $n$.
We now apply our technique developed in Section~\ref{sec:fourier} to prove the following upper bound 
for  the granularity of low-degree polynomials.
\begin{theorem}\label{thm:granularity_degree}
For any Boolean function $f: \BntB$, if $d = \deg_2(f)$ is the $\gf$-degree of $f$,
then 
$\gran(f^{\pm}) \leq n-\lceil \frac{n}{d} \rceil$, and 
consequently, $\gran(f) \leq n-\lceil \frac{n}{d} \rceil+1$. 
\end{theorem}
\begin{proof}
Suppose $\widehat{f^{\pm}}(T)$ achieves $\gran(f^{\pm})$, i.e., 
$\widehat{f^{\pm}}(T) = c/2^{\scriptsize \gran(f^{\pm})}$ for some odd integer $c$. 
Without loss of generality, we may assume that $T \neq \emptyset$. 
Actually, if $\widehat{f^{\pm}}(0)$ is the single Fourier coefficient that achieves $\gran(f^{\pm})$, 
then the sum of the squares of all Fourier coefficients of 
$f^{\pm}$ would be a rational number with granularity $2\gran(f^{\pm})$ instead of $1$,
contradicting Parseval's identity.
	
Now we apply an invertible linear map $L$ such that $(L^T)^{-1}(T)=[n]$.
Denote $f\circ L$ by $g$. By Fact~\ref{fact:invertible1}, $g$ is also a polynomial of degree $d$.
Moreover, by Fact~\ref{fact:invertible2}, 
we have that $\widehat{g^{\pm}}([n]) = \widehat{f^{\pm}}(T)$. 

Now suppose $g(x) = \sum_{i=1}^m \prod_{j \in S_{i}}x_{j}$,
where $|S_j|\leq d$ for every $1\leq j \leq m$.
Applying Lemma~\ref{lem:main} and notice that, 
since $|S_j|\leq d$, the minimum number $k$ such that there exists a collection of $k$ subsets 
from $\{S_j\}_{j\in [m]}$ that cover $[n]$ is $k=\lceil \frac{n}{d} \rceil$.
Therefore, by Eq.~\eqref{eq:Fourier_main},
\[
\widehat{g^{\pm}}([n]) = (-1)^{n} w_g([n]) = (-1){^n} \sum_{j=k}^{m}\frac{(-2)^{j} N_{j}([n])}{2^n}.
\]
Note that the granularity of the $j^{\text{th}}$ term in the above summation is at most $n-j$ (we only have inequality
here as $N_{j}([n])$ may be an even number), and the granularity of a sum of rational numbers is at most the
maximum granularity in the summands:
\[
\gran\left(\sum_{j=1}^{\ell}y_{j}\right) \leq \max_{1\leq j\leq \ell}\gran(y_{j}),
\] 
where $y_{j} \in \Q$ for $1\leq j \leq \ell$, 
we therefore have $\gran\left(\widehat{g^{\pm}}([n])\right) \leq n-k = n-\lceil \frac{n}{d} \rceil$. 
This finally gives
\[
\gran(f^{\pm}) = \gran\left(\widehat{f^{\pm}}(T)\right) = \gran\left(\widehat{g^{\pm}}([n])\right) 
\leq n - \lceil \frac{n}{d} \rceil.
\]
The upper bound of the granularity of $f$ follows from the easy fact that $\gran(f)\leq \gran(f^{\pm})+1$.
\end{proof}

\section*{Acknowledgements}
We are indebted to the anonymous reviewers for their detailed helpful comments.

\newcommand{\etalchar}[1]{$^{#1}$}



\begin{thebibliography}{TWXZ13}

\bibitem[AFH12]{AFH12}
Anil Ada, Omar Fawzi, and Hamed Hatami.
\newblock Spectral norm of symmetric functions.
\newblock In {\em Proceedings of the 15th International Workshop on
  Approximation, Randomization, and Combinatorial Optimization}, pages
  338--349, 2012.

\bibitem[BC99]{BC99}
Anna Bernasconi and Bruno Codenotti.
\newblock Spectral analysis of boolean functions as a graph eigenvalue problem.
\newblock {\em IEEE Transactions on Computers}, 48(3):345--351, 1999.

\bibitem[CT13]{CT13}
Gil Cohen and Avishay Tal.
\newblock Two structural results for low degree polynomials and applications.
\newblock {\em ECCC}, TR13-145, 2013.

\bibitem[Fin47]{Fin47}
Nathan Fine.
\newblock Binomial coefficients modulo a prime.
\newblock {\em American Mathematical Monthly}, 54:589--592, 1947.

\bibitem[GOS{\etalchar{+}}11]{GOS+11}
Parikshit Gopalan, Ryan O'Donnell, Rocco Servedio, Amir Shpilka, and Karl
  Wimme.
\newblock Testing {F}ourier dimensionality and sparsity.
\newblock {\em SIAM Journal on Computing}, 40(4):1075--1100, 2011.

\bibitem[KN97]{KN97}
Eyal Kushilevitz and Noam Nisan.
\newblock {\em Communication Complexity}.
\newblock Cambridge University Press, Cambridge, UK, 1997.

\bibitem[LLZ11]{LLZ11}
Ming~Lam Leung, Yang Li, and Shengyu Zhang.
\newblock Tight bounds on the communication complexity of symmetric {XOR}
  functions in one-way and {SMP} models.
\newblock In {\em Proceedings of the 8th Annual Conference on Theory and
  Applications of Models of Computation}, pages 403--408, 2011.

\bibitem[Lov14a]{Lov14}
Shachar Lovett.
\newblock Communication is bounded by root of rank.
\newblock In {\em Proceedings of the 46th Annual ACM Symposium on Theory of
  Computing}, pages 842--846, 2014.

\bibitem[Lov14b]{Lov14a}
Shachar Lovett.
\newblock Recent advances on the log rank conjecture.
\newblock In {\em Bulletin of {EATCS}}, 2014.

\bibitem[LS88]{LS88}
L{\'a}szl{\'o} Lov{\'a}sz and Michael~E. Saks.
\newblock Lattices, {M}{\"o}bius functions and communication complexity.
\newblock In {\em Proceedings of the 29th Annual Symposium on Foundations of
  Computer Science}, pages 81--90, 1988.

\bibitem[LZ10]{LZ10}
Troy Lee and Shengyu Zhang.
\newblock Composition theorems in communication complexity.
\newblock In {\em Proceedings of the 37th International Colloquium on Automata,
  Languages and Programming}, pages 475--489, 2010.

\bibitem[LZ13]{LZ13}
Yang Liu and Shengyu Zhang.
\newblock Quantum and randomized communication complexity of {XOR} functions in
  the {SMP} model.
\newblock {\em ECCC}, 20(10), 2013.

\bibitem[MO09]{MO09}
Ashley Montanaro and Tobias Osborne.
\newblock On the communication complexity of {XOR} functions, 2009.
\newblock http://arxiv.org/abs/0909.3392v2.

\bibitem[NW95]{NW95}
Noam Nisan and Avi Wigderson.
\newblock On rank vs. communication complexity.
\newblock {\em Combinatorica}, 15(4):557--565, 1995.

\bibitem[STV14]{STV14}
Amir Shpilka, Avishay Tal, and Ben~Lee Volk.
\newblock On the structure of boolean functions with small spectral norm.
\newblock In {\em Proceedings of the 5th Innovations in Theoretical Computer
  Science}, 2014.

\bibitem[TWXZ13]{TWXZ13}
Hing~Yin Tsang, Chung~Hoi Wong, Ning Xie, and Shengyu Zhang.
\newblock Fourier sparsity, spectral norm, and the log-rank conjecture.
\newblock In {\em Proceedings of the 54th Annual {IEEE} Symposium on
  Foundations of Computer Science}, pages 658--667, 2013.

\bibitem[Yao79]{Yao79}
Andrew Yao.
\newblock Some complexity questions related to distributive computing.
\newblock In {\em Proceedings of the 11th Annual ACM Symposium on Theory of
  Computing}, pages 209--213, 1979.

\bibitem[ZS09]{ZS09}
Zhiqiang Zhang and Yaoyun Shi.
\newblock Communication complexities of symmetric {XOR} functions.
\newblock {\em Quantum Information {\&} Computation}, 9(3):255--263, 2009.

\end{thebibliography}
\end{document}